\documentclass[11pt,a4paper]{article}
\usepackage{amsfonts}
\usepackage{amsmath}
\usepackage{amssymb}
\usepackage{graphicx}

\newtheorem{acknowledgement}{Acknowledgement}

\newtheorem{definition}{Definition}

\newtheorem{proposition}{Proposition}

\newenvironment{proof}[1][Proof]{\textbf{#1.} }{\  \rule{0.5em}{0.5em}}
\makeatletter
\def \@removefromreset#1#2{\let \@tempb \@elt
     \def \@tempa#1{@&#1}\expandafter \let \csname @*#1*\endcsname \@tempa
     \def \@elt##1{\expandafter \ifx \csname @*##1*\endcsname \@tempa \else
    \noexpand \@elt{##1}\fi}     \expandafter \edef \csname cl@#2\endcsname{\csname cl@#2\endcsname}     \let \@elt \@tempb
     \expandafter \let \csname @*#1*\endcsname \@undefined}

\@removefromreset{equation}{section}

\@removefromreset{theorem}{section}
\makeatother
\input{tcilatex}
\begin{document}

\title{Quantifying tolerance of a nonlocal \\
multi-qudit state to any local noise}
\author{Elena R. Loubenets \\
Applied Mathematics Department, \\
National Research University Higher School of Economics, \\
Moscow, 101000, Russia}
\maketitle

\begin{abstract}
We present a general approach for quantifying tolerance of a nonlocal $N$%
-partite state to any local noise under different classes of quantum
correlation scenarios with arbitrary numbers of settings and outcomes at
each site. This allows us to derive new precise bounds in $d$ and $N$ on
noise tolerances for: (i) an arbitrary nonlocal $N$-qudit state; (ii) the $N$%
-qudit Greenberger-Horne-Zeilinger (GHZ) state; (iii) the $N$-qubit $W$
state and the $N$-qubit Dicke states, and to analyse asymptotics of these
precise bounds for large $N$ and $d.$
\end{abstract}

\section{Introduction}

Nonlocality \cite{1, 01, 02} of an $N$-qudit quantum state, \emph{in the
sense} \emph{of its} \emph{violation of a Bell inequality, }is a major
resource for developing quantum information technologies. Conceptual and
quantitative issues of Bell nonlocality in a general nonsignaling case have
been analyzed in \cite{3} and references therein. The main concepts and
tools which were developed to describe and to study Bell nonlocality in a
quantum case have been reviewed in \cite{2}. (We further discuss only the
notions of Bell nonlocality and locality and, therefore, mostly suppress the
specification "Bell" before these terms.)

In quantum information applications, one, however, deals with noisy channels
and, for a nonlocal $N$-qudit state $\rho _{d,N}$ $,$ $d\geq 2,N\geq 2,$ it
is important to evaluate amounts of noise not breaking the nonclassical
character of its statistical correlations. Analytical and numerical bounds
on the critical visibility\footnote{%
For the rigorous definition of this notion, see Section 4.} of a nonlocal $N$%
-qudit state $\rho _{d,N}$ in a mixture \emph{with white noise:} 
\begin{equation}
\left( 1-\beta \right) \frac{\mathbb{I}^{\otimes N}}{d^{N}}+\beta \rho
_{d,N},\text{ \ \ }\beta \in \lbrack 0,1],  \label{1}
\end{equation}%
have been intensively studied in the literature: (i) for a nonlocal
two-qudit state -- in \cite{4, 5, 6, 7, 8} and references therein; (ii) for
some specific quantum correlation scenarios and specific $N$-qubit states --
in \cite{9, 10, 11, 12, 13, 14, 15, 16, 17}, and (iii) for an arbitrary
nonlocal $N$-qudit state $\rho _{d,N},$ $N\geq 3,d\geq 3$ -- in \cite{18}.

However, precise analytical bounds on the critical visibility of a nonlocal $%
N$-qudit state $\rho _{d,N}$ in a mixture \ 
\begin{equation}
\left( 1-\beta \right) \zeta _{loc}+\beta \rho _{d,N},\text{ \ \ }\beta \in
\lbrack 0,1],  \label{2}
\end{equation}%
with an arbitrary local noise\footnote{%
That is, a noise described by a local $N$-qudit state $\zeta _{loc}.$} and,
more generally, bounds on the \emph{tolerance}\footnote{%
For the rigorous definition of this notion, see Section 4.} of a nonlocal $N$%
-qudit state $\rho _{d,N}$ to any local noise are not, to our knowledge,
known in a general $N$-qudit case, though for a nonlocal family of joint
probabilities under a bipartite ($N=2$) correlation scenario, the similar
concept -- the resistance to noise -- was introduced in \cite{19} and
further discussed in \cite{2}.

We note that, for many quantum information applications based on Bell
nonlocality, it is important to evaluate the maximal amount of noise
tolerable by a nonlocal $N$-qudit state and this amount is determined
specifically via the noise tolerance of a nonlocal state.

In the present paper, due to the general framework for Bell nonlocality
developed in \cite{20, 21, 3}, we present a consistent approach to
quantifying tolerance of a nonlocal $N$-partite quantum state to any local
noise under different classes of quantum correlation scenarios with
arbitrary numbers of settings and any spectral types of outcomes at each
site. This allows us:

\begin{itemize}
\item to specify via parameters of an $N$-partite state the general
analytical expressions for the noise tolerance of a nonlocal $N$-partite
state (i) under $S_{1}\times \cdots \times S_{N}$-setting quantum
correlation scenarios with any number of outcomes at each site and (ii)
under all quantum correlation scenarios with arbitrary numbers of settings
and outcomes per site;

\item to derive \emph{new} precise lower/upper bounds in $d$ and $N$ on the
noise tolerances and the maximal amounts of tolerable local noise for: (i)
an arbitrary nonlocal $N$-qudit state; (ii) the $N$-qudit GHZ state; (iii)
the $N$-qubit $W$ state and the $N$-qubit Dicke states and to analyse
asymptotics of these precise new bounds for large $N$ and $d.$
\end{itemize}

\section{General $N$-partite Bell inequalities}

Let us shortly recall the notion of a general multipartite Bell inequality 
\cite{23} with arbitrary numbers of settings and outcomes per site. For the
general framework on the probabilistic description of an arbitrary
multipartite correlation scenario with any number of settings and any
spectral type of outcomes at each site, see \cite{24}.

Consider a correlation scenario, where each $n$-th of $N$ parties performs $%
S_{n}\geq 1$ measurements with outcomes $\lambda _{n}\in \lbrack -1,1]$ and
every measurement at $n$-th site is specified by a positive integer $%
s_{n}=1,...,S_{n}$. For concreteness, we label an $S_{1}\times \cdots \times
S_{N}$-setting scenario by $\mathcal{E}_{S},$ where $S=S_{1}\times \cdots
\times S_{N}.$

For a correlation scenario $\mathcal{E}_{S}$, denote by $%
P_{(s_{1},...,s_{N})}^{(\mathcal{E}_{S})}$ the joint probability
distribution of outcomes $(\lambda _{1},\ldots ,\lambda _{N})\in \lbrack
-1,1]^{N}$ under an $N$-partite joint measurement induced by measurements $%
s_{1},...,s_{N}$ at the corresponding sites and by 
\begin{eqnarray}
\mathcal{B}_{\Phi _{S}}^{(\mathcal{E}_{S})}
&=&\sum_{s_{1},...,s_{_{N}}}\left\langle f_{(s_{1},...,s_{N})}(\lambda
_{1},\ldots ,\lambda _{N})\right\rangle _{\mathcal{E}_{S}},  \label{3} \\
\Phi _{S} &=&\{f_{(s_{1},...,s_{N})}:[-1,1]^{N}\rightarrow \mathbb{R}\mid
s_{n}=1,...,S_{n},\text{ \ }n=1,...,N\},  \notag
\end{eqnarray}%
\emph{a} \emph{linear combination} of averages (expectations)%
\begin{eqnarray}
&&\left\langle f_{(s_{1},...,s_{N})}(\lambda _{1},\ldots ,\lambda
_{N})\right\rangle _{\mathcal{E}_{S}}  \label{4} \\
&=&\int\limits_{[-1,1]^{N}}f_{(s_{1},...,s_{N})}(\lambda _{1},\ldots
,\lambda _{N})P_{(s_{1},...,s_{N})}^{(\mathcal{E}_{S})}\left( \mathrm{d}%
\lambda _{1}\times \cdots \times \mathrm{d}\lambda _{N}\right) ,  \notag
\end{eqnarray}%
of the most general form, specified for each $N$-partite joint measurement $%
(s_{1},...,s_{N})$ by a bounded real-valued function $f_{(s_{1},...,s_{N)}}$
of outcomes $\left( \lambda _{1},\ldots ,\lambda _{N}\right) $ $\in $ $%
[-1,1]^{N}$ at all $N$ sites. Each linear combination (\ref{3}) is specified
by a family $\Phi _{S}=\{f_{(s_{1},...,s_{N})}\}$ of these functions.

Depending on a choice of a function $f_{(s_{1},...,s_{N})}$, an average (\ref%
{4}) may refer either to the joint probability of events observed under this
joint measurement at $M\leq N$ sites or to the expectation 
\begin{equation}
{\big \langle}\lambda _{1}^{(s_{1})}\cdot \ldots \cdot \lambda
_{n_{M}}^{(s_{n_{M}})}{\big \rangle}_{\mathcal{E}_{S}}=\int%
\limits_{[-1,1]^{N}}\lambda _{1}\cdot \ldots \cdot \lambda
_{n_{M}}P_{(s_{1},...,s_{N})}^{(\mathcal{E}_{S})}\left( \mathrm{d}\lambda
_{1}\times \cdots \times \mathrm{d}\lambda _{N}\right)   \label{5}
\end{equation}%
of the product of outcomes observed at $M\leq N$ sites or may have a more
complicated form. In quantum information, the product expectation (\ref{5})
is referred to as a correlation function.

Let the probabilistic description of a correlation scenario $\mathcal{E}_{S}$
admit\footnote{%
For the general statements on the LHV\ modelling, see section 4 in \cite{24}.%
} \emph{a local hidden variable (LHV) model, }that is,\emph{\ }all joint
probability distributions $\left\{ P_{(s_{1},...,s_{N})}^{(\mathcal{E}%
_{S})},s_{n}=1,...,S_{n},n=1,...,N\right\} $ of this scenario admit the
representation%
\begin{eqnarray}
&&P_{(s_{1},...,s_{N})}^{(\mathcal{E}_{S})}\left( \mathrm{d}\lambda
_{1}\times \cdots \times \mathrm{d}\lambda _{N}\right)  \label{6} \\
&=&\dint\limits_{\Omega }P_{1,s_{1}}(\mathrm{d}\lambda _{1}|\omega )\cdot
\ldots \cdot P_{N,s_{_{N}}}(\mathrm{d}\lambda _{N}|\omega )\nu _{\mathcal{E}%
_{S}}(\mathrm{d}\omega )  \notag
\end{eqnarray}%
in terms of a single probability distribution $\nu _{\mathcal{E}_{S}}(%
\mathrm{d}\omega )$ of some variables $\omega \in \Omega $ and conditional
probability distributions $P_{n,s_{n}}(\mathrm{\cdot }|\omega ),$ referred
to as "local" in the sense that each $P_{n,s_{n}}(\mathrm{\cdot }|\omega )$
at $n$-th site depends only on the corresponding measurement $%
s_{n}=1,...,S_{n}$ at this site.

In this case, each linear combination (\ref{3}) of scenario averages
satisfies the tight LHV constraints \cite{23}: 
\begin{equation}
\mathcal{B}_{\Phi _{S}}^{\inf }\leq \mathcal{B}_{\Phi _{S}}^{(\mathcal{E}%
_{S})}{\big |}_{_{lhv}}\leq \mathcal{B}_{\Phi _{S}}^{\sup }  \label{7}
\end{equation}%
with the LHV constants%
\begin{eqnarray}
\mathcal{B}_{\Phi _{S}}^{\sup } &=&\sup_{\lambda _{n}^{(s_{n})}\in \lbrack
-1,1],\forall s_{n},\forall n}\text{ }%
\sum_{s_{1},...,s_{_{N}}}f_{(s_{1},...,s_{N})}(\lambda _{1}^{(s_{1})},\ldots
,\lambda _{N}^{(s_{N})}),  \label{8} \\
\mathcal{B}_{\Phi _{S}}^{\inf } &=&\inf_{\lambda _{n}^{(s_{n})}\in \lbrack
-1,1],\forall s_{n},\forall n}\text{ }%
\sum_{s_{1},...,s_{_{N}}}f_{(s_{1,}...,s_{N})}(\lambda _{1}^{(s_{1})},\ldots
,\lambda _{N}^{(s_{N})}).  \notag
\end{eqnarray}%
From (\ref{7}) it follows that, in the LHV\ case, 
\begin{equation}
\left\vert \text{ }\mathcal{B}_{\Phi _{S}}^{(\mathcal{E}_{S})}{\big |}%
_{_{lhv}}\right\vert \leq \mathcal{B}_{\Phi _{S}}^{lhv}=\max \left\{
\left\vert \mathcal{B}_{\Phi _{S}}^{\sup }\right\vert ,\left\vert \mathcal{B}%
_{\Phi _{S}}^{\inf }\right\vert \right\} .  \label{9}
\end{equation}%
Some of the LHV inequalities in (\ref{7}) may be fulfilled for a wider (than
LHV) class of correlation scenarios. This is, for example, the case for the
LHV\ constraints on joint probabilities following explicitly from
nonsignaling of probability distributions. Moreover, some of the LHV
inequalities in (\ref{7}) may be simply trivial, i. e. fulfilled for
correlation scenarios of all types, not necessarily nonsignaling. (For the
latter general concept and its relation to the EPR (Einstein-Podolsky-Rosen)
locality and Bell locality, see Sections 2, 3 in \cite{24}.)

\emph{Each of the tight LHV inequalities in (\ref{7}) that may be violated
under a non-LHV scenario is referred to as a (general) }$S_{1}\times \cdots
\times S_{N}$\emph{-setting} \emph{Bell inequality.}

\section{Quantum violation}

Let an $S_{1}\times \cdots \times S_{N}$-setting correlation scenario be
performed on a quantum state $\rho $ on a complex Hilbert space $\mathcal{H}%
_{1}\otimes \cdots \otimes \mathcal{H}_{N}$. For this correlation scenario,
every $N$-partite joint measurement $(s_{1},...,s_{N})$ is described by the
joint probability distribution 
\begin{equation}
\mathrm{tr}[\rho \{\mathrm{M}_{1,s_{1}}(\mathrm{d}\lambda _{1})\otimes
\cdots \otimes \mathrm{M}_{N,s_{N}}(\mathrm{d}\lambda _{N})\}],  \label{10}
\end{equation}%
where each $\mathrm{M}_{n,s_{n}}(\mathrm{\cdot })$ is a normalized positive
operator-valued (\emph{POV}) measure, representing on $\mathcal{H}_{n}$ a
generalized quantum measurement $s_{n}$ at $n$-th site. For concreteness, we
further specify this quantum correlation scenario by symbol $\mathcal{E}%
_{\rho ,\mathrm{M}_{S}}$, where $\mathrm{M}_{S}=\{\mathrm{M}_{n,s_{n}}\}$ is
a collection of POV measures at all $N$ sites.

Since the probabilistic description of a quantum correlation scenario does
not need \cite{1, 01, 02} to admit an LHV model, in a quantum case, Bell
inequalities may be violated. The parameter \cite{20}%
\begin{equation}
\mathrm{\Upsilon }_{S_{1}\times \cdots \times S_{N}}^{(\rho )}=\sup_{_{\Phi
_{S},\text{ }\mathrm{M}_{S}}}\frac{1}{\mathcal{B}_{\Phi _{S}}^{lhv}}%
\left\vert \mathcal{B}_{\Phi _{S}}^{(\mathcal{E}_{\rho ,\mathrm{M}%
_{S}})}\right\vert \geq 1  \label{11}
\end{equation}%
specifies the maximal violation by an $N$-partite quantum state $\rho $ of
all general $S_{1}\times \cdots \times S_{N}$-setting Bell inequalities for
any number of outcomes at each site and the parameter \cite{20} 
\begin{equation}
\mathrm{\Upsilon }_{\rho }=\sup_{S_{1},...,S_{N}}\mathrm{\Upsilon }%
_{S_{1}\times \cdots \times S_{N}}^{(\rho )}\geq 1  \label{12}
\end{equation}%
-- the maximal violation by an $N$-partite quantum state $\rho $ of all
general Bell inequalities for any numbers of settings and outcomes at each
site.

From (\ref{3}), (\ref{4}), (\ref{10})--(\ref{12}) it follows that, for any
convex mixture $\rho =\sum \gamma _{i}\rho _{i},$ $\gamma _{i}\geq 0,$ $%
\sum_{i}\gamma _{i}=1,$ 
\begin{eqnarray}
1 &\leq &\mathrm{\Upsilon }_{S_{1}\times \cdots \times S_{N}}^{(\rho )}\leq
\sum_{i}\gamma _{i}\mathrm{\Upsilon }_{S_{1}\times \cdots \times
S_{N}}^{(\rho _{i})},  \label{13} \\
1 &\leq &\mathrm{\Upsilon }_{\rho }\leq \sum_{i}\gamma _{i}\mathrm{\Upsilon }%
_{\rho _{i}}.  \label{14}
\end{eqnarray}

\begin{definition}
An $N$-partite quantum state $\rho $ is referred to as $S_{1}\times \cdots
\times S_{N}$-setting nonlocal \cite{18} if it violates an $S_{1}\times
\cdots \times S_{N}$-setting Bell inequality and overall nonlocal (or simply
nonlocal) if it violates any of Bell inequalities.
\end{definition}

Clearly, an $S_{1}\times \cdots \times S_{N}$-setting nonlocal state is
(overall) nonlocal\emph{\ }but not vice versa. From Definition 1, (\ref{11}%
), (\ref{12}) and Proposition 6 in \cite{20} it follows that an $N$-partite
quantum state $\rho $ is \cite{20, 18}:

\begin{itemize}
\item $S_{1}\times \cdots \times S_{N}$-setting nonlocal iff 
\begin{equation}
\mathrm{\Upsilon }_{S_{1}\times \cdots \times S_{N}}^{(\rho )}>1  \label{15}
\end{equation}%
and $S_{1}\times \cdots \times S_{N}$-setting local iff 
\begin{equation}
\mathrm{\Upsilon }_{S_{1}\times \cdots \times S_{N}}^{(\rho )}=1.
\label{15_1}
\end{equation}

\item (overall) nonlocal iff 
\begin{equation}
\mathrm{\Upsilon }_{\rho }>1  \label{16}
\end{equation}%
and fully Bell local \cite{18} iff 
\begin{equation}
\mathrm{\Upsilon }_{\rho }=1.  \label{16_1}
\end{equation}
\end{itemize}

For details and the one-to-one correspondence of relations (\ref{15_1}), (%
\ref{16_1}) to the LHV modelling of the corresponding quantum correlation
scenarios on an $N$-partite quantum state $\rho ,$ see Sections 5 and 6 in 
\cite{20}.

\section{Tolerance to any local noise}

Let $\rho $ be a nonlocal quantum state on $\mathcal{H}_{1}\otimes \cdots
\otimes \mathcal{H}_{N}$ and $S_{1},...,S_{N}$ be arbitrary numbers of
measurement settings at the corresponding parties' sites. Denote by 
\begin{equation}
\beta _{S_{1}\times \cdots \times S_{N}}^{(\rho )}(\zeta _{loc})\in (0,1]
\label{17}
\end{equation}%
\emph{the }$S_{1}\times \cdots \times S_{N}$\emph{-setting critical
visibility }of a nonlocal $N$-partite state $\rho $ in a convex mixture with
noise described by a local $N$-partite state\emph{\ }$\zeta _{loc}:$ 
\begin{equation}
\left( 1-\beta \right) \zeta _{loc}\text{ }+\beta \rho ,\text{ \ \ }\beta
\in \lbrack 0,1].  \label{18}
\end{equation}%
In terminology specified in Section 3, the threshold $\beta _{S_{1}\times
\cdots \times S_{N}}^{(\rho )}(\zeta _{loc})$ means that a noisy state (\ref%
{18}) is $S_{1}\times \cdots \times S_{N}$-setting nonlocal iff 
\begin{equation}
\beta \in {\large (}\beta _{S_{1}\times \cdots \times S_{N}}^{(\rho )}(\zeta
_{loc}),1{\large ]}
\end{equation}%
and $S_{1}\times \cdots \times S_{N}$-setting local iff 
\begin{equation}
\beta \in {\large [}0,\beta _{S_{1}\times \cdots \times S_{N}}^{(\rho
)}(\zeta _{loc}){\large ]}.
\end{equation}

If $\beta _{S_{1}\times \cdots \times S_{N}}^{(\rho )}(\zeta _{loc})=1,$
then a noisy state (\ref{18}) is $S_{1}\times \cdots \times S_{N}$-setting
local for all $\beta \in \lbrack 0,1].$ For $\beta =1,$ the latter implies
that though an $N$-partite state $\rho $ is (overall) nonlocal, it does not
violate any of $S_{1}\times \cdots \times S_{N}$-setting Bell inequalities,
i.e. this state $\rho $ is $S_{1}\times \cdots \times S_{N}$-setting local.

Let $\mathcal{L}_{S_{1}\times \cdots \times S_{N}}^{(nonloc)}$ be the set of
all $S_{1}\times \cdots \times S_{N}$-setting nonlocal $N$-partite states on 
$\mathcal{H}_{1}\otimes \cdots \otimes \mathcal{H}_{N}$ and $\mathcal{L}%
_{N}^{(nonloc)}\mathcal{\supset L}_{S_{1}\times \cdots \times
S_{N}}^{(nonloc)}$ -- the set of all (overall) nonlocal $N$-partite states.

\begin{definition}
For a nonlocal $N$-partite state $\rho \in \mathcal{L}_{N}^{(nonloc)},$ we
call 
\begin{equation}
\mathfrak{T}_{S_{1}\times \cdots \times S_{N}}^{(\rho )}=\sup_{\zeta
_{loc}}\beta _{S_{1}\times \cdots \times S_{N}}^{(\rho )}(\zeta _{loc})\in
(0,1]  \label{20}
\end{equation}%
the $S_{1}\times \cdots \times S_{N}$-setting tolerance to any local noise.
Otherwise expressed, 
\begin{equation}
\mathfrak{T}_{S_{1}\times \cdots \times S_{N}}^{(\rho )}=\inf \{\beta \in
\lbrack 0,1]\mid \left( 1-\beta \right) \zeta _{loc}\text{ }+\beta \rho \in 
\mathcal{L}_{S_{1}\times \cdots \times S_{N}}^{(nonloc)},\text{ }\forall
\zeta _{loc}\}.  \label{21}
\end{equation}
\end{definition}

Clearly, a noisy state (\ref{18}) is $S_{1}\times \cdots \times S_{N}$%
-setting nonlocal for \emph{any} local noise iff $\beta \in (\mathfrak{T}%
_{S_{1}\times \cdots \times S_{N}}^{(\rho )},1]$.

If $\mathfrak{T}_{S_{1}\times \cdots \times S_{N}}^{(\rho )}=1,$ then a
nonlocal $N$-partite state $\rho \in \mathcal{L}_{N}^{(nonloc)}$ is $%
S_{1}\times \cdots \times S_{N}$-setting local. Since, however, $\rho $ is
overall nonlocal, there exist numbers $\widetilde{S}_{1},...,\widetilde{S}%
_{N}$ of measurement settings at the corresponding sites for which this
state is $\widetilde{S}_{1},...,\widetilde{S}_{N}$-setting nonlocal. For
these settings, $\mathfrak{T}_{\widetilde{S}_{1}\times \cdots \times 
\widetilde{S}_{N}}^{(\rho )}<1.$

\begin{definition}
For a nonlocal $N$-partite state $\rho \in \mathcal{L}_{N}^{(nonloc)},$ we
call%
\begin{equation}
\mathfrak{T}_{\rho }=\inf_{S_{1},...,S_{N}}\mathfrak{T}_{S_{1}\times \cdots
\times S_{N}}^{(\rho )}\in (0,1)  \label{22}
\end{equation}%
the overall tolerance (or simply tolerance) to any local noise.
\end{definition}

This definition implies that, for all $\beta \in (\mathfrak{T}_{\rho },1]$,
a noisy state (\ref{18}) specified for a nonlocal state $\rho \in \mathcal{L}%
_{N}^{(nonloc)}$ is nonlocal for any local noise.\smallskip

\emph{The smaller is the value of the noise tolerance }$\mathfrak{T}_{\rho }$%
\emph{\ of a nonlocal N-partite state }$\rho $\emph{, the greater is the
maximal amount }%
\begin{equation}
\mathfrak{M}_{\rho }=1-\mathfrak{T}_{\rho }  \label{22'}
\end{equation}%
\emph{of a local noise of any type tolerable\footnote{%
In the sense that a noisy state (\ref{18}) specified for a nonlocal state $%
\rho $ is also nonlocal.} by this nonlocal state under all quantum
correlation scenarios and, therefore, the greater is the robustness of
nonlocality of a state }$\rho $ \emph{to any local noise.}

\begin{proposition}
Let $\rho $ be a nonlocal N-partite state and $S_{1},...,S_{N}$ -- arbitrary
numbers of measurement settings at the corresponding sites. The $S_{1}\times
\cdots \times S_{N}$-setting tolerance of a nonlocal state $\rho $ to any
local noise has the form 
\begin{equation}
\mathfrak{T}_{S_{1}\times \cdots \times S_{N}}^{(\rho )}=\frac{2}{1+\mathrm{%
\Upsilon }_{S_{1}\times \cdots \times S_{N}}^{(\rho )}}  \label{23}
\end{equation}%
and the overall noise tolerance of a nonlocal state $\rho $ is given by 
\begin{equation}
\mathfrak{T}_{\rho }=\inf_{S_{1},...,S_{N}}\mathfrak{T}_{S_{1}\times \cdots
\times S_{N}}^{(\rho )}=\frac{2}{1+\mathrm{\Upsilon }_{\rho }},  \label{24}
\end{equation}%
where $\mathrm{\Upsilon }_{S_{1}\times \cdots \times S_{N}}^{(\rho )}$ is
the maximal violation (\ref{11}) by a nonlocal state $\rho $ of all $%
S_{1}\times \cdots \times S_{N}$-setting general Bell inequalities and $%
\mathrm{\Upsilon }_{\rho }$ -- the maximal violation (\ref{12}) by a
nonlocal state $\rho $ of all general Bell inequalities.
\end{proposition}

\begin{proof}
From (\ref{3}), (\ref{4}), (\ref{11}) and linearity in $\rho $ of quantum
probability distributions (\ref{10}) it follows%
\begin{equation}
\beta \mathrm{\Upsilon }_{S_{1}\times \cdots \times S_{N}}^{(\rho )}\leq 
\mathrm{\Upsilon }_{S_{1}\times \cdots \times S_{N}}^{(\left( 1-\beta
\right) \zeta _{loc}\text{ }+\beta \rho )}+(1-\beta )\mathrm{\Upsilon }%
_{S_{1}\times \cdots \times S_{N}}^{(\zeta _{loc})}.\text{\ }  \label{25}
\end{equation}%
By Definition 2, for each $\beta \in \lbrack 0,\mathfrak{T}_{S_{1}\times
\cdots \times S_{N}}^{(\rho )}],$ there exists $\widetilde{\zeta }_{loc}$
such that a noisy state (\ref{18}) is $S_{1}\times \cdots \times S_{N}$%
-setting local. For this $\widetilde{\zeta }_{loc},$ relation (\ref{15_1})\
implies%
\begin{equation}
\mathrm{\Upsilon }_{S_{1}\times \cdots \times S_{N}}^{(\left( 1-\beta
\right) \widetilde{\zeta }_{loc}\text{ }+\beta \rho )}=1.  \label{26}
\end{equation}%
Also, $\mathrm{\Upsilon }_{S_{1}\times \cdots \times S_{N}}^{(\widetilde{%
\zeta }_{loc})}=1.$ Taking this into account in (\ref{25}), (\ref{26}) we
have 
\begin{eqnarray}
\beta \mathrm{\Upsilon }_{S_{1}\times \cdots \times S_{N}}^{(\rho )} &\leq &%
\mathrm{\Upsilon }_{S_{1}\times \cdots \times S_{N}}^{(\left( 1-\beta
\right) \widetilde{\zeta }_{loc}\text{ }+\beta \rho )}+(1-\beta )\mathrm{%
\Upsilon }_{S_{1}\times \cdots \times S_{N}}^{(\widetilde{\zeta }%
_{loc})}=2-\beta \text{\ \ \ \ }  \label{27} \\
&\Leftrightarrow &\text{ \ \ }\beta \leq \frac{2}{1+\mathrm{\Upsilon }%
_{S_{1}\times \cdots \times S_{N}}^{(\rho )}}  \notag
\end{eqnarray}%
for each $\beta \in \lbrack 0,\mathfrak{T}_{S_{1}\times \cdots \times
S_{N}}^{(\rho )}]$. Therefore,%
\begin{equation}
\mathfrak{T}_{S_{1}\times \cdots \times S_{N}}^{(\rho )}\leq \frac{2}{1+%
\mathrm{\Upsilon }_{S_{1}\times \cdots \times S_{N}}^{(\rho )}}.  \label{28}
\end{equation}

On the other hand, for each $\beta >\mathfrak{T}_{S_{1}\times \cdots \times
S_{N}}^{(\rho )},$ a noisy state (\ref{18}) is $S_{1}\times \cdots \times
S_{N}$-setting nonlocal for every $\zeta _{loc},$ so that by (\ref{15}) the
relation%
\begin{equation}
\mathrm{\Upsilon }_{S_{1}\times \cdots \times S_{N}}^{(\left( 1-\beta
\right) \zeta _{loc}\text{ }+\beta \rho )}>1  \label{29}
\end{equation}%
holds for all $\zeta _{loc}$ and each $\beta >\mathfrak{T}_{S_{1}\times
\cdots \times S_{N}}^{(\rho )}.$ In view of (\ref{11}) and linearity in $%
\rho $ of quantum probability distributions (\ref{10}), this, in turn,
implies that, for each $\zeta _{loc}$ and every $\beta >\mathfrak{T}%
_{S_{1}\times \cdots \times S_{N}}^{(\rho )},$ there exist (i) an $%
S_{1}\times \cdots \times S_{N}$-setting Bell inequality, specified in (\ref%
{7}) by some family $\widetilde{\Phi }_{S}=\Phi _{S}(\zeta _{loc},\beta )$
of functions in (\ref{4}) and (ii) quantum measurements, specified by some
family $\widetilde{\mathrm{M}}_{S}=\mathrm{M}_{S}(\zeta _{loc},\beta )$ of
POV measures, such that%
\begin{equation}
\left\vert (1-\beta )\mathcal{B}_{\widetilde{\Phi }_{S}}^{(\mathcal{E}%
_{\zeta _{loc},\widetilde{\mathrm{M}}_{S}})}+\beta \mathcal{B}_{\widetilde{%
\Phi }_{S}}^{(\mathcal{E}_{\rho ,\widetilde{\mathrm{M}}_{S}})}\right\vert >%
\mathcal{B}_{\widetilde{\Phi }_{S}}^{lhv}  \label{30}
\end{equation}%
for all $\zeta _{loc}$ and all $\beta \in (\mathfrak{T}_{S_{1}\times \cdots
\times S_{N}}^{(\rho _{d,N})},1].$ Varying (\ref{30}) in $\zeta ,$ implies
that, for each $x:=$ $\mathcal{B}_{\widetilde{\Phi }_{S}}^{(\mathcal{E}%
_{\zeta _{loc},\widetilde{\mathrm{M}}_{S}})}/\mathcal{B}_{\widetilde{\Phi }%
_{S}}^{lhv}\in \lbrack -1,1],$ there must exist $y:=\mathcal{B}_{\widetilde{%
\Phi }_{S}}^{(\mathcal{E}_{\rho ,\widetilde{\mathrm{M}}_{S}})}/\mathcal{B}_{%
\widetilde{\Phi }_{S}}^{lhv}\in $ $[-\mathrm{\Upsilon }_{S_{1}\times \cdots
\times S_{N}}^{(\rho )},\mathrm{\Upsilon }_{S_{1}\times \cdots \times
S_{N}}^{(\rho )}]$ such that 
\begin{equation}
\left\vert \text{ }(1-\beta )x+\beta y\right\vert >1  \label{31}
\end{equation}%
for all $\beta \in (\mathfrak{T}_{S_{1}\times \cdots \times S_{N}}^{(\rho
)},1].$ But this is possible if 
\begin{equation}
\frac{2-\beta }{\beta }<\mathrm{\Upsilon }_{S_{1}\times \cdots \times
S_{N}}^{(\rho )}\text{ \ \ \ }\Leftrightarrow \text{ \ \ \ }\beta >\frac{2}{%
1+\mathrm{\Upsilon }_{S_{1}\times \cdots \times S_{N}}^{(\rho )}}.
\label{32}
\end{equation}%
Together with condition that (\ref{31}) holds for all $\beta >\mathfrak{T}%
_{S_{1}\times \cdots \times S_{N}}^{(\rho )}$ and definition (\ref{21}) of
the tolerance $\mathfrak{T}_{S_{1}\times \cdots \times S_{N}}^{(\rho )}$,
Eq. (\ref{32}) implies%
\begin{equation}
\mathfrak{T}_{S_{1}\times \cdots \times S_{N}}^{(\rho )}\geq \frac{2}{1+%
\mathrm{\Upsilon }_{S_{1}\times \cdots \times S_{N}}^{(\rho )}}.  \label{33}
\end{equation}%
Inequalities (\ref{28}), (\ref{33}) prove relation (\ref{23}). Expression (%
\ref{24}) follows from (\ref{23}) and (\ref{12}).
\end{proof}

For a bipartite correlation scenario with two settings per site, expression (%
\ref{23}) for the $S_{1}\times \cdots \times S_{N}$-setting noise tolerance
of a nonlocal $N$-qudit quantum state agrees with minimizing over all
possible parties' measurements on a two-qudit state of the resistance of a
nonlocal family of scenario joint probability distributions to any local
noise -- the notion that was introduced in \cite{19}.

\section{General bounds}

Let us evaluate due to (\ref{23}), (\ref{24}) the noise tolerances for an
arbitrary nonlocal state $\rho _{d,N}$ on $(\mathbb{C}^{d})^{\otimes N}$.

In view of our results in \cite{20, 21, 22} we have the following general
precise upper bounds on the maximal violation $\mathrm{\Upsilon }_{S\times
\cdots \times S}^{(\rho _{d,N})}$ of all $S\times \cdots \times S$-setting
Bell inequalities by an arbitrary nonlocal $N$-qudit state $\rho _{d,N},$ $%
d\geq 2,$ $N\geq 2:$ 
\begin{eqnarray}
\mathrm{\Upsilon }_{2\times \cdots \times 2}^{(\rho _{d,N})} &\leq &\min {%
\big \{}d^{\frac{N-1}{2}},\text{ }3^{N-1}{\big \}},  \label{34} \\
\mathrm{\Upsilon }_{S\times \cdots \times S}^{(\rho _{d,N})} &\leq &\min {%
\big \{}d^{\frac{S(N-1)}{2}},\text{ }\left( 2\min \{d,S\}-1\right) ^{N-1}{%
\big \}},\text{ \ \ }S\geq 3,  \label{35}
\end{eqnarray}%
under projective quantum measurements at all sites and%
\begin{equation}
\mathrm{\Upsilon }_{S\times \cdots \times S}^{(\rho _{d,N})}\leq \left(
2\min \{d,S\}-1\right) ^{N-1},\text{ \ \ }S\geq 2,  \label{36}
\end{equation}%
under generalized quantum measurements at all sites.

The maximal violation $\mathrm{\Upsilon }_{\rho _{d,N}}$ by a nonlocal $N$%
-qudit state $\rho _{d,N}$ of \emph{all} general Bell inequalities satisfies
the relation \cite{22}%
\begin{equation}
\mathrm{\Upsilon }_{\rho _{d,N}}\leq \left( 2d-1\right) ^{N-1},\text{ \ \ }%
d\geq 2,N\geq 2  \label{37}
\end{equation}%
for either of above types of quantum measurements.

Taking these upper bounds into account in expressions (\ref{23}), (\ref{24}%
), we derive the following bounds on the $S\times \cdots \times S$-setting
noise tolerance $\mathfrak{T}_{S\times \cdots \times S}^{(\rho _{d,N})}$ of
an arbitrary nonlocal $N$-qudit state $\rho _{d,N},$ $d\geq 2,$ $N\geq 2:$ 
\begin{eqnarray}
\mathfrak{T}_{2\times \cdots \times 2}^{(\rho _{d,N})} &\geq &\frac{2}{%
1+\min {\big \{}d^{\frac{N-1}{2}},\text{ }3^{N-1}{\big \}}},  \label{38} \\
&&  \notag \\
\mathfrak{T}_{S\times \cdots \times S}^{(\rho _{d,N})} &\geq &\frac{2}{%
1+\min {\big \{}d^{\frac{S(N-1)}{2}},\left( 2\min \{d,S\}-1\right) ^{N-1}{%
\big \}}},\text{ \ }S\geq 3,  \notag
\end{eqnarray}%
under projective quantum measurements at all sites and 
\begin{equation}
\mathfrak{T}_{S\times \cdots \times S}^{(\rho _{d,N})}\geq \frac{2}{1+\left(
2\min \{d,S\}-1\right) ^{N-1}},\text{ \ \ }S\geq 2,  \label{41}
\end{equation}%
under generalized quantum measurements at all sites.

The overall noise tolerance $\mathfrak{T}_{\rho _{d,N}}$ of an arbitrary
nonlocal $N$-qudit state $\rho _{d,N}$ satisfies the relation 
\begin{equation}
\mathfrak{T}_{\rho _{d,N}}\geq \frac{2}{1+\left( 2d-1\right) ^{N-1}},\text{
\ \ \ }d\geq 2,\text{ }N\geq 2,  \label{42}
\end{equation}%
for either of above types of quantum measurements. For $d\rightarrow \infty
, $ this lower bound decreases to zero as $\frac{2}{\left( 2d\right) ^{N-1}}%
. $

From (\ref{41}) it follows that, for an arbitrary $S\times \cdots \times S$%
-setting nonlocal $N$-qudit state $\widetilde{\rho }_{d,N}$, the maximal
amount 
\begin{equation}
\mathfrak{M}_{S\times \cdots \times S}^{(\widetilde{\rho }_{d,N})}=1-%
\mathfrak{T}_{S\times \cdots \times S}^{(\widetilde{\rho }_{d,N})}
\label{42'}
\end{equation}%
of tolerable local noise is upper bounded by 
\begin{equation}
\mathfrak{M}_{S\times \cdots \times S}^{(\widetilde{\rho }_{d,N})}\leq \frac{%
\left( 2\min \{d,S\}-1\right) ^{N-1}-1}{\left( 2\min \{d,S\}-1\right)
^{N-1}+1},\text{ \ \ }d\geq 2,\text{ }N\geq 2,\text{ }S\geq 2,  \label{k3}
\end{equation}%
under all $S\times \cdots \times S$-setting quantum correlation scenarios.

For $d\leq S,$ this upper bound does not depend on a number $S$ of
measurement settings per site: 
\begin{equation}
\mathfrak{M}_{S\times \cdots \times S}^{(\widetilde{\rho }_{d,N})}\leq \frac{%
\left( 2d-1\right) ^{N-1}-1}{\left( 2d-1\right) ^{N-1}+1},\text{ \ \ }N\geq
2,\text{ \ }S\geq d\geq 2,  \label{k1}
\end{equation}%
while, for $d\geq S,$ it does not depend on a qudit dimension $d:$%
\begin{equation}
\mathfrak{M}_{S\times \cdots \times S}^{(\widetilde{\rho }_{d,N})}\leq \frac{%
\left( 2S-1\right) ^{N-1}-1}{\left( 2S-1\right) ^{N-1}+1},\text{ \ \ }N\geq
2,\text{ \ }d\geq S\geq 2.  \label{k2}
\end{equation}

For example, for two-qudit and three-qudit cases and two measurement
settings per site, the general bound (\ref{k2}) implies the following upper
bounds on the maximal tolerable noise: 
\begin{equation}
\mathfrak{M}_{2\times 2}^{(\widetilde{\rho }_{d,2})}\leq \frac{1}{2},\text{
\ \ \ }\mathfrak{M}_{2\times 2\times 2}^{(\widetilde{\rho }_{d,3})}\leq 
\frac{4}{5},  \label{k_3}
\end{equation}%
for all dimensions $d\geq 2.$

\section{N-qudit GHZ state}

Consider now bounds on the noise tolerances for the $N$-qudit
Greenberger-Horne-Zeilinger (GHZ)\ state 
\begin{equation}
GHZ_{d,N}=\frac{1}{d}\sum_{j,j_{1}}\left( |e_{j}\rangle \langle
e_{j_{1}}|\right) ^{\otimes N}  \label{y}
\end{equation}%
on $\left( \mathbb{C}^{d}\right) ^{\otimes N}.$ Here, $\left\{
e_{m},m=1,...,d\right\} $ is an orthonormal base in $\mathbb{C}^{d}.$

For this $N$-qudit quantum state, the maximal Bell violation $\mathrm{%
\Upsilon }_{S\times \cdots \times S}^{(GHZ_{d,N})}$ of all $S\times \cdots
\times S$-setting Bell inequalities admits the upper bounds \cite{20, 21, 22}
which are more specific than the general bounds (\ref{34})--(\ref{37}).
Namely, for all $d\geq 2,$ $N\geq 2:$%
\begin{eqnarray}
\mathrm{\Upsilon }_{2\times \cdots \times 2}^{(GHZ_{d,N})} &\leq &\min {\big
\{}d^{\frac{N-1}{2}},\text{ }3^{N-1},1+2^{N-1}(d-1){\big \}},  \label{x} \\
&&  \notag \\
\mathrm{\Upsilon }_{S\times \cdots \times S}^{(GHZ_{d,N})} &\leq &\min {\big
\{}d^{\frac{S(N-1)}{2}},(2S-1)^{N-1},1+2^{N-1}(d-1){\big \}},\text{ \ \ }%
S\geq 3,  \notag
\end{eqnarray}%
under projective quantum measurements and 
\begin{equation}
\mathrm{\Upsilon }_{S\times \cdots \times S}^{(GHZ_{d,N})}\leq \min {\big \{}%
(2S-1)^{N-1},1+2^{N-1}(d-1){\big \}},\text{ \ \ }S\geq 2,  \label{x3}
\end{equation}%
under generalized quantum measurements.

The maximal violation $\mathrm{\Upsilon }_{GHZ_{d,N}}$ by the GHZ state of
all general Bell inequalities satisfies the relation \cite{20}: 
\begin{equation}
\mathrm{\Upsilon }_{GHZ_{d,N}}\leq 1+2^{N-1}(d-1),\text{ \ \ }d\geq 2,\text{ 
}N\geq 2,  \label{44}
\end{equation}%
for either of above types of quantum measurements.

In view of (\ref{x})--(\ref{44}), for the $N$-qudit GHZ state, the general
bounds (\ref{38})--(\ref{k_3}) on the noise tolerances and the maximal
amount of tolerable noise can be improved.

Taking (\ref{x})--(\ref{44}) into account in (\ref{23}), (\ref{24}), we come
to the following bounds for the $S\times \cdots \times S$-setting noise
tolerance of the $N$-qudit GHZ state for all $d\geq 2,$ $N\geq 2:$%
\begin{eqnarray}
\mathfrak{T}_{2\times \cdots \times 2}^{(GHZ_{d,N})} &\geq &\frac{2}{1+\min {%
\big \{}d^{\frac{N-1}{2}},3^{N-1},1+2^{N-1}(d-1){\big \}}},  \label{43_1} \\
&&  \notag \\
\mathfrak{T}_{S\times \cdots \times S}^{(GHZ_{d,N})} &\geq &\frac{2}{1+\min {%
\big \{}d^{\frac{S(N-1)}{2}},\left( 2S-1\right) ^{N-1},1+2^{N-1}(d-1){\big \}%
}},\text{ \ \ }S\geq 3,  \notag
\end{eqnarray}%
under projective quantum measurements and 
\begin{equation}
\mathfrak{T}_{S\times \cdots \times S}^{(GHZ_{d,N})}\geq \frac{2}{1+\min %
\big \{\left( 2S-1\right) ^{N-1},1+2^{N-1}(d-1)\big\}},\text{ \ \ }S\geq 2,
\label{43_4}
\end{equation}%
under generalized quantum measurements.

The overall noise tolerance of the $N$-qudit state GHZ state satisfies the
relation 
\begin{equation}
\mathfrak{T}_{GHZ_{d,N}}\geq \frac{1}{1+2^{N-2}(d-1)},\text{ \ }d\geq 2,%
\text{ }N\geq 2,\text{\ }  \label{45}
\end{equation}%
for either of above types of quantum measurements.

From (\ref{43_1}) it follows that, under $2\times \cdots \times 2$-setting
correlation scenarios\emph{\ with projective measurements} at all sites, the
maximal amount (\ref{42'}) of a local noise of any type tolerable by the $N$%
-qudit GHZ state is upper bounded by%
\begin{equation}
\mathfrak{W}_{2\times \cdots \times 2}^{(GHZ_{d,N})}\leq \min \big \{\frac{%
d^{\frac{N-1}{2}}-1}{d^{\frac{N-1}{2}}+1},\frac{3^{N-1}-1}{3^{N-1}+1},\frac{%
2^{N-2}(d-1)}{2^{N-2}(d-1)+1}\big\}  \label{45'}
\end{equation}%
for all $d\geq 2,N\geq 2.$

For example, for the three-qutrit GHZ state ($N=3,d=3$)$,$ the maximal
amount of tolerable local noise $\mathfrak{W}_{2\times 2\times
2}^{(GHZ_{3,3})}\leq \frac{2}{3}.$

By (\ref{45}), under all quantum correlation scenarios, the maximal amount (%
\ref{22'}) of any local noise tolerable by the $N$-qudit GHZ state satisfies
the relation%
\begin{equation}
\mathfrak{W}_{GHZ_{d,N}}\leq \frac{2^{N-2}(d-1)}{2^{N-2}(d-1)+1},\text{ \ \ }%
d\geq 2,\text{ }N\geq 2.
\end{equation}%
For the two-qudit GHZ state ($N=2$), this general bound gives%
\begin{equation}
\mathfrak{W}_{GHZ_{d,2}}\leq \frac{d-1}{d},\text{ \ \ }d\geq 2.  \label{45__}
\end{equation}

\subsection{N-qubit case}

Let us evaluate the noise tolerances for the $N$-qubit GHZ state\emph{.}

As it has been proven in \cite{21}, under projective measurements at all
sites, the maximal violation $\mathrm{\Upsilon }_{2\times \cdots \times
2}^{(GHZ_{2,N})}$ by the $N$-qibit state $GHZ_{2,N}$ of \emph{all general} $%
2\times \cdots \times 2$\emph{-}setting Bell inequalities coincides with the
maximal violation $2^{\frac{N-1}{2}}$ by this state of all \emph{correlation}
$2\times \cdots \times 2$-setting\emph{\ }Bell inequalities and is,
therefore, given by 
\begin{equation}
\mathrm{\Upsilon }_{2\times \cdots \times 2}^{(GHZ_{2,N})}\mid
_{proj.meas}=2^{\frac{N-1}{2}}.  \label{46}
\end{equation}

By (\ref{23}) this implies that, under $2\times \cdots \times 2$-setting
quantum correlation scenarios \emph{with projective quantum measurements} at
all sites, the $2\times \cdots \times 2$-setting noise tolerance of the $N$%
-qubit GHZ\ state is given by%
\begin{equation}
\mathfrak{T}_{2\times \cdots \times 2}^{(GHZ_{2,N})}\mid _{proj.meas}=\frac{2%
}{1+2^{\frac{N-1}{2}}},\text{ \ \ }N\geq 2.  \label{47}
\end{equation}%
while the maximal amount (\ref{45'}) of any local noise tolerable by the $N$%
-qubit GHZ state is 
\begin{equation}
\mathfrak{W}_{2\times \cdots \times 2}^{(GHZ_{2,N})}\mid _{proj.meas}=\frac{%
2^{\frac{N-1}{2}}-1}{2^{\frac{N-1}{2}}+1},\text{ \ \ }N\geq 2.
\end{equation}%
For $N=2,3,$ this bound implies%
\begin{equation}
\mathfrak{W}_{2\times \cdots \times 2}^{(GHZ_{2,2})}\mid _{proj.meas}=\frac{%
\sqrt{2}-1}{\sqrt{2}+1},\text{ \ \ \ \ }\mathfrak{W}_{2\times \cdots \times
2}^{(GHZ_{2,3})}\mid _{proj.meas}=\frac{1}{3}.
\end{equation}

Furthermore, in view of (\ref{12}), (\ref{44}) and (\ref{46}), the maximal
violation $\mathrm{\Upsilon }_{GHZ_{2,N}}$ by the $N$-qubit GHZ state of 
\emph{all general} Bell inequalities satisfies the relations 
\begin{equation}
2^{\frac{N-1}{2}}\leq \mathrm{\Upsilon }_{GHZ_{2,N}}\leq 1+2^{N-1}
\label{48}
\end{equation}%
under all generalized parties' quantum measurements.

Therefore, by (\ref{24}) the overall noise tolerance $\mathfrak{T}%
_{GHZ_{2,N}}$ of the $N$-qubit GHZ\ state admits the bounds 
\begin{equation}
\frac{1}{1+2^{N-2}}\leq \mathfrak{T}_{GHZ_{2,N}}\leq \frac{2}{1+2^{\frac{N-1%
}{2}}},\text{ \ \ }N\geq 2,  \label{49}
\end{equation}%
where the lower and upper bounds decrease with increasing $N.$ This, in
particular, implies that, under all quantum correlation scenarios, a mixture
of the $N$-qubit state $GHZ_{2,N}$ \emph{with any local noise} is nonlocal
for all 
\begin{equation}
\beta >\frac{2}{1+2^{\frac{N-1}{2}}}\underset{N>>1}{\simeq }2^{-\frac{N-3}{2}%
}.  \label{as1}
\end{equation}%
For comparison: a mixture of the $N$-qubit GHZ state \emph{with white noise}
is nonlocal \cite{15} for all $\beta >2^{-\frac{N-1}{2}}.$

Due to (\ref{49}), we derive the following bounds for the maximal amount $%
\mathfrak{M}_{GHZ_{2,N}}$ of a local noise of any type tolerable by the the $%
N$-qubit GHZ state:%
\begin{equation}
\frac{2^{\frac{N-1}{2}}-1}{2^{\frac{N-1}{2}}+1}\leq \mathfrak{M}%
_{GHZ_{2,N}}\leq \frac{2^{N-2}}{1+2^{N-2}},\text{ \ \ \ }N\geq 2.
\label{49'}
\end{equation}

Hence, \emph{with increasing of a number }$N$ \emph{of qubits, the
robustness of nonlocality of the }$N$\emph{-qubit GHZ state to any local
noise increases. }

\section{N-qubit Dicke states}

In this section, we evaluate the overall noise tolerance for the $N$-qubit
Dicke states $|D_{N}^{(k)}\rangle \langle D_{N}^{(k)}|$ on $\left( \mathbb{C}%
^{2}\right) ^{\otimes N}$ with $k=1,...,N-1$ excitations: 
\begin{equation}
|D_{N}^{(k)}\rangle =\frac{1}{\sqrt{\binom{N}{k}}}\sum_{j}\pi _{j}(|0\rangle
^{\otimes (N-k)}\otimes |1\rangle ^{\otimes k}).  \label{51}
\end{equation}%
Here, $\{|0\rangle ,|1\rangle \}$ is a orthonormal base in $\mathbb{C}^{2},$
notation $\pi _{j}$ means a permutation in the tensor product of $k$ vectors 
$|1\rangle $ and $(N-k)$ vectors $|0\rangle $ and the binomial coefficient $%
\binom{N}{k}$ gives the number of such permutations.

For example, for $k=1$, the Dicke state $|D_{N}^{(1)}\rangle $ constitutes
the $N$-qubit $W$ state 
\begin{eqnarray}
|W_{N}\rangle &=&\frac{1}{\sqrt{N}}(\underset{N}{\underbrace{|0\rangle
\otimes \cdots \otimes |0\rangle \otimes |1\rangle }}+\underset{N}{%
\underbrace{|0\rangle \otimes \cdots \otimes |1\rangle \otimes |0\rangle }}
\label{52} \\
&&+...+\underset{N}{\underbrace{|1\rangle \otimes \cdots \otimes |0\rangle
\otimes |0\rangle }}).  \notag
\end{eqnarray}%
For $N=3,$ $k=2$, the three-qubit Dicke state with two excitations has the
form%
\begin{equation}
|D_{3}^{(2)}\rangle =\frac{1}{\sqrt{3}}(|0\rangle \otimes |1\rangle \otimes
|1\rangle +|1\rangle \otimes |0\rangle \otimes |1\rangle +|1\rangle \otimes
|1\rangle \otimes |0\rangle ).  \label{53}
\end{equation}

Due to the general upper bound (\ref{37}) and the value of violation by the
Dicke state $|D_{N}^{(k)}\rangle $ of the specific Bell inequality
introduced in \cite{15}, the maximal violation (\ref{12}) by the Dicke state 
$|D_{N}^{(k)}\rangle $ of \emph{all} general Bell inequalities satisfies the
relation 
\begin{equation}
1+\frac{2^{N-1}(\sqrt{2}-1)}{\binom{N}{k}}\leq \mathrm{\Upsilon }%
_{D_{N}^{(k)}}\leq 3^{N-1},\text{ \ \ }N\geq 2,\text{ }k=1,...,N-1.
\label{54}
\end{equation}%
From (\ref{24}), (\ref{54}) it follows that, for the Dicke state $%
|D_{N}^{(k)}\rangle ,$ the overall noise tolerance $\mathfrak{T}%
_{D_{N}^{(k)}}$ admits the bounds 
\begin{equation}
\frac{2}{1+3^{N-1}}\leq \mathfrak{T}_{D_{N}^{(k)}}\leq \frac{1}{1+\frac{%
2^{N-2}(\sqrt{2}-1)}{\binom{N}{k}}}  \label{55}
\end{equation}%
for all $N\geq 2,$ $k=1,2,...,N-1.$

Therefore, under all quantum correlation scenarios, a mixture (\ref{18}) of
the $N$-qubit Dicke state $|D_{N}^{(k)}\rangle $ with any local noise is
nonlocal for all 
\begin{equation}
\beta >\frac{1}{1+\frac{2^{N-2}(\sqrt{2}-1)}{\binom{N}{k}}}.  \label{56}
\end{equation}%
For a large even $N>>1$ and $k=\frac{N}{2},$ the binomial coefficient $%
\binom{N}{N/2}\underset{N>>1}{\simeq }\frac{2^{N}\sqrt{2}}{\sqrt{\pi N}},$
so that by (\ref{56}) a mixture of the Dicke state $|D_{N}^{(\frac{N}{2}%
)}\rangle $ with any local noise is nonlocal for all%
\begin{equation}
\beta >\frac{1}{1+\frac{2^{N-2}(\sqrt{2}-1)}{\binom{N}{N/2}}}\underset{N>>1}{%
\simeq }\frac{4\sqrt{2}}{\left( \sqrt{2}-1\right) }\frac{1}{\sqrt{\pi N}}.
\label{57}
\end{equation}

Specifying (\ref{55}) for $k=1,$ we have the following lower and upper
bounds for the noise tolerance $\mathfrak{T}_{W_{N}}$ of the $N$-qubit $W$
state (\ref{52}): 
\begin{equation}
\frac{2}{1+3^{N-1}}\leq \mathfrak{T}_{W_{N}}\leq \frac{N}{N+2^{N-2}(\sqrt{2}%
-1)},\text{ \ \ }N\geq 2.  \label{55'}
\end{equation}%
Hence, a mixture (\ref{18}) of the $N$-qubit $W$\ state with any local noise
is nonlocal for all 
\begin{equation}
\beta >\frac{N}{N+2^{N-2}(\sqrt{2}-1)},\text{ \ \ \ }N\geq 2.  \label{56'}
\end{equation}

\section{Conclusions}

In this article, we have presented a new general framework for quantifying
noise tolerances of a nonlocal $N$-partite quantum state under different
classes of quantum correlation scenarios with arbitrary numbers of settings
and outcomes at each site.

This allowed us (i) to consistently specify two types (\ref{20}), (\ref{22})
of noise tolerances of a nonlocal $N$-partite qudit state; (ii) to express
them due to (\ref{23}), (\ref{24}) in terms of the maximal violations (\ref%
{11}), (\ref{12}) by this state of two classes of general Bell inequalities
and (iii) to derive further the following new precise lower and upper bounds
on the\ noise tolerances and the maximal amounts of tolerable local noise:

\begin{itemize}
\item bounds (\ref{38})--(\ref{k_3}) -- for an arbitrary nonlocal $N$-qudit
state;

\item bounds (\ref{43_1})--(\ref{45__}), (\ref{47}), (\ref{49})--(\ref{49'})
-- for the $N$-qudit GHZ state, in particular, the $N$-qubit GHZ state;

\item bounds (\ref{55})--(\ref{56'}) -- for the $N$-qubit Dicke states and
the $N$-qubit $W$ state.
\end{itemize}

\noindent\ and to analyze their asymptotics for large $N$ and $d$. We, in
particular, prove that with increasing of a number $N$ of qubits, the
robustness of nonlocality of the $N$-qubit GHZ state to any local noise
increases.\smallskip

To our knowledge, no one of these analytical bounds has been reported in the
literature.

\begin{acknowledgement}
Valuable discussions with Professor Khrennikov are very much appreciated.
The publication was prepared within the framework of the Academic Fund
Program at the National Research University Higher School of Economics (HSE)
in 2018-2019 (grant N 18-01-0064) and by the Russian Academic Excellence
Project "5-100".
\end{acknowledgement}

\end{document}